\documentclass[sigconf]{acmart}
\pdfoutput=1

\renewcommand\footnotetextcopyrightpermission[1]{} 
\usepackage{graphicx,verbatim,multirow,dcolumn}
\usepackage{amsmath,amssymb,bm}
\usepackage{amsthm}
\usepackage{algorithmicx}
\usepackage[ruled,vlined]{algorithm2e}
\usepackage{booktabs}
\usepackage{comment}
\usepackage[capitalise]{cleveref}
\usepackage{microtype}
\usepackage{lettrine}
\usepackage{xcolor}
\usepackage{array}
\usepackage{graphicx}
\usepackage{float}
\usepackage{tikz}
\usepackage{bm}
\usepackage{url}
\usepackage{subcaption}
\usepackage{placeins}
\usepackage{afterpage}
\usepackage{pifont}
\usepackage{mathtools}
\usepackage{epstopdf}

\usepackage[aboveskip=1pt,skip=1pt,belowskip=1pt]{caption}

\newcommand{\multigraph}{K}
\newcommand{\xhdr}[1]{\vspace{0.5mm}\noindent{{\bf #1.}}}

\newcommand{\xmark}{\ding{55}}%

\newcommand{\duration}[1]{\Delta(#1)}
\newcommand{\given}{\;\vert\;}

\crefname{figure}{Figure}{Figures}

\newcommand{\phantomsubfigure}[1]{\begin{subfigure}[b]{0.1\textwidth}\phantomcaption\label{#1}\end{subfigure}}

\newtheorem{problem}{Problem}

\newcommand{\E}{\mathbb{E}}
\newcommand{\Var}{\mathrm{Var}}

\setcopyright{rightsretained}

\acmDOI{xx.xxx/xxx_x}
\acmISBN{xxx-xxxx-xx-xxx/xx/xx}
\acmConference[WSDM'19]{ACM WSDM conference}{February 2019}{Melbourne, Australia}
\acmYear{2019}
\copyrightyear{2019}
\acmSubmissionID{xxx-xxx-xx}

\begin{document}

\title{A sampling framework for counting temporal motifs}

\author{Paul Liu}
\affiliation{%
  \institution{Stanford University}
  \city{Stanford}
  \state{California}
}
\email{paul.liu@stanford.edu}

\author{Austin R. Benson}
\affiliation{%
  \institution{Cornell University}
  \city{Ithaca}
  \state{New York}
}
\email{arb@cs.cornell.edu}

\author{Moses Charikar}
\affiliation{%
  \institution{Stanford University}
  \city{Stanford}
  \state{California}
}
\email{moses@cs.stanford.edu}


\date{\today}


\begin{abstract}
Pattern counting in graphs is fundamental to network science tasks, and
there are many scalable methods for approximating counts of small
patterns, often called motifs, in large graphs.
However, modern graph datasets now contain richer structure, and incorporating
temporal information in particular has become a critical part of network analysis.
Temporal motifs, which are generalizations of small subgraph
patterns that incorporate temporal ordering on edges, are an emerging part of
the network analysis toolbox.
However, there are no algorithms for fast estimation of temporal motifs counts;
moreover, we show that even counting simple temporal star motifs is NP-complete.
Thus, there is a need for fast and approximate algorithms.
Here, we present the first frequency estimation algorithms for counting temporal
motifs.
More specifically, we develop a sampling framework that sits as a layer on top
of existing exact counting algorithms and enables fast and accurate memory-efficient estimates of
temporal motif counts.
Our results show that we can achieve one to two orders of magnitude speedups with
minimal and controllable loss in accuracy on a number of datasets.
\end{abstract}

\maketitle


\section{Scalable pattern counting in temporal network data}

Pattern counting is one of the fundamental problems in data mining~\cite{Chakrabarti-2006-curriculum,Han-2011-book}.
A particularly important case is counting patterns in graph data,
which is used within a variety of network analysis tasks such as
anomaly detection~\cite{Noble-2003-anomaly,Sun-2007-GraphScope},
role discovery~\cite{Henderson-2012-RolX,Rossi-2015-role},
and clustering~\cite{Rohe-2013-blessing,Benson-2016-hoo,Tsourakakis-2017-scalable}.
These methods typically make use of features derived from
the frequencies of small graph patterns---usually
called motifs~\cite{Milo-2002-motifs} or graphlets~\cite{Przulj-2004-graphlet}
(we adopt the ``motif'' terminology in this paper)---and
are used across a range of disciplines, including
social network analysis~\cite{Leskovec-2010-signed,Ugander-2013-subgraphs},
neuroscience~\cite{Hu-2012-motifs,Battiston-2017-motif}, and
computational biology~\cite{Mangan-2003-ffl,Przulj-2007-comparison}.
Furthermore, the counts of motifs have also been used to automatically uncover
fundamental design principles in complex systems~\cite{Milo-2002-motifs,Mangan-2003-ffl,Milo-2004-superfamilies}.

The scale of graph datasets has led to a number of algorithms for estimating the
frequency of motif
counts~\cite{Ahmed-2014-graphsample,Elenberg-2016-profiles,Bressan-2017-graphlets,Jain-2017-cliques,Wang-2018-moss}.
For example, just the task of estimating the number of triangles in a graph has
garnered a substantial amount of
attention~\cite{Tsourakakis-2009-doulion,Avron-2010-counting,Lim-2015-mascot,Seshadhri-2013-wedge,Eden-2017-triangles,Stefani-2017-triest}.
Many of these algorithms are based on sampling procedures amenable to streaming
models of graph data~\cite{Feigenbaum-2005-streaming,McGregor-2014-streaming}.
At this point, there is a reasonably mature set of algorithmic and statistical
tools available for approximately counting motifs in large graph datasets.

While graphs have become large enough to warrant frequency estimation
algorithms, graph datasets have, at the same time, become richer in structure.
A particularly important type of rich information is
time~\cite{Kossinets-2008-pathways,Holme-2012-temporal,Farajtabar-2015-coevolve,Gaumont-2016-dense,Scholtes-2017-networks}.
Specifically, in this paper, we consider datasets where edges are accompanied by
a timestamp, such as the time a transaction was made with a cryptocurrency,
the time an email was sent between colleagues, or the time a packet was forwarded
from one IP address to another by a router. Accordingly, motifs have been
generalized to incorporate temporal
information~\cite{Zhao-2010-communication-motifs,Kovanen-2011-motifs,Paranjape-2017-motifs}
and have already been used in a variety of
applications~\cite{Lahiri-2007-structure,Shao-2013-temporal,Meydan-2013-prediction,Kovanen-2013-motifs}.
However, we do not yet have algorithmic tools for estimating
frequencies of temporal motifs in these large temporal graphs. This is
especially problematic since including timestamps increases the size of the
stored data; for example, a traditional email graph would only record
if one person \emph{has ever} emailed another person, whereas the temporal version
of the same network would record \emph{every time} there is a communication.

To exacerbate the problem, counting temporal motifs turns out to be
fundamentally more difficult in a computational complexity sense. In
particular, we prove that counting basic temporal star motifs is NP-complete.
This contrasts sharply with stars in traditional static graphs, which
are generally considered trivial to count (the number of non-induced $k$-edge
stars with center node $u$ is simply ${d \choose k}$, where $d$ is the degree of
$u$). Thus, our result highlights how counting problems in temporal graphs
involve fundamentally more challenging computations, thus further motivating the
need for approximation algorithms.

Here we develop the first frequency estimation algorithms for counting temporal
motifs. We focus on the definition of temporal motifs from Paranjape et
al.~\cite{Paranjape-2017-motifs}, but our methodology is general and could be
adapted for other definitions. Our approach is based on sampling that employs as
a subroutine any algorithm (satisfying some mild conditions) that \emph{exactly}
counts the number of instances of temporal motifs. Thus, our methodology
provides a way to accelerate existing
algorithms~\cite{Paranjape-2017-motifs,Mackey-2018-chronological}, as well as
better exact counting algorithms that could be developed in the future.

At a basic level, our sampling framework partitions time into intervals, uses
some algorithm to find exact motif counts in a subset of the intervals, and
weights these counts to get an estimate of the number of temporal motifs.
A key challenge is that the time duration of a temporal motifs can cross interval
boundaries, which makes it challenging to obtain an accurate frequency estimator
since motifs of larger duration are more likely to be omitted.
At its core, our sampling framework uses importance sampling~\cite{mcbook} in two different
ways. First, we use importance sampling as a way to design an unbiased estimator
by appropriately scaling the exact counts appearing in some intervals. Second,
we use importance sampling as a way to (probabilistically) choose which
intervals to sample, which reduces the variance of our unbiased estimator.

In addition to the scalability advantages offered by sampling, our framework has
two other important features. First, the sampling requires a smaller amount of
memory. We show an example where this enables us to count a
complex motif on a large temporal graph when existing exact counting algorithms run out
memory. Second, the sampling procedure has built-in opportunity for
parallel computation, which provides a path to faster computation with exact counting
algorithms that do not have built-in parallelism.

As discussed above, our sampling framework employs an exact counting algorithm
as a subroutine. The constraints on the algorithm are that it must provide the
exact counts along with the so-called \emph{duration} of the motif instance (the
difference in the earliest and latest timestamp in the edges in the motif
instance; for example, the duration in the top left motif instance in
\cref{fig:example_C} is 32 - 16 = 16). This constraint holds for some existing
algorithms~\cite{Mackey-2018-chronological} but not for
others~\cite{Paranjape-2017-motifs}.
An additional contribution of our work is
a new exact counting algorithm for a class of star motifs that is compatible
with our sampling framework. As an added bonus, this new exact counting
algorithm actually out-performs existing algorithms.

We test our sampling procedure on several temporal graph datasets from a variety
of domains, ranging in size from 60,000 to over 600 million temporal edges and
find that our sampling framework can improve run time by one to two orders
of magnitude while maintaining a relative error tolerance of 5\% in the counts.
The variance analysis of our error bounds tends to be pessimistic, since we make
no assumptions on the distribution of timestamps within our datasets. Thus, we
also show empirically that our worst-case bounds are far from what we see in the
data.


\section{Preliminaries on temporal motifs}

We first review some basic notions of temporal motifs. There are a
few types of temporal motifs, which we discuss in the context of
related work in \cref{sec:related}.
Here we review the definitions used by Paranjape et
al.~\cite{Paranjape-2017-motifs}, which is one of the more flexible definitions
that also poses difficult computational challenges.

A \emph{temporal edge} is a timestamped directed edge between an ordered pair of
nodes. A collection of temporal edges is a \emph{temporal graph} (see
\cref{fig:example_A}). Formally, a temporal graph $T$ on a node set $V$ is a
collection of tuples $(u_i, v_i, t_i)$, $i = 1, \ldots, m$, where each $u_i$ and
$v_i$ are elements of $V$ and each $t_i$ is a timestamp in $\mathbb{R}$. There
can be many temporal edges from $u$ to $v$ (one example is in email data, where
one person sends an email to another many times). We assume that the timestamps
$t_i$ are unique so that the temporal edges in a graph can be ordered. This
assumption makes the presentation of the paper simpler, but our methods can
handle temporal graphs with non-unique timestamps.

\begin{figure}[t]
\phantomsubfigure{fig:example_A}
\phantomsubfigure{fig:example_B}
\phantomsubfigure{fig:example_C}
\centering
\includegraphics[width=\linewidth]{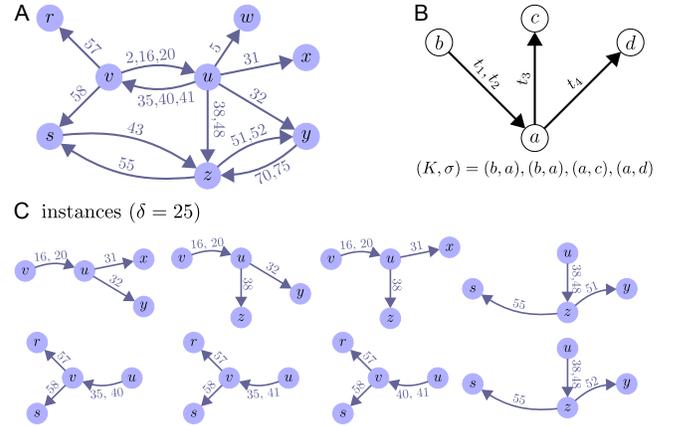}
\caption{%
Temporal graph and temporal motifs.
(A) Illustration of a temporal graph. The numbers along edges correspond to
timestamps. There can be multiple timestamped edges between a given pair of
nodes.
(B) Illustration of a motif, which is formally a multigraph $\multigraph$ with an ordering $\sigma$
on its edges.
(C) Eight $\delta$-instances of the motif in the temporal graph with
$\delta=25$.  The motifs match the multigraph, the edge ordering, and appear
within the time span $\delta$. The sequence of temporal edges $(u,v,16)$,
$(u,v,20)$, $(u,y,32)$, $(u,z,48)$ is \emph{not} a $\delta$-instance of the motif
because all edges do not fit within the time span $\delta$. The duration of a motif
instance $M'$, denoted $\duration{M'}$, is the difference between the last and
first timestamps; for example, the duration of the instance in the top left is
$32-16=16$.
}
\label{fig:example}
\end{figure}

If we ignore time and duplicate edges, the temporal graph induces a
standard (static) directed graph. Formally, the \emph{static graph}
of a temporal graph $T$ on a node set $V$ is a graph $G = (V, E)$,
where $E = \{ (u, v) \given \exists t : (u, v, t) \in T\}$.
Edges in $G$ are called \emph{static edges}.

Next, we formalize temporal motifs (illustrated in \cref{fig:example_B}).
\begin{definition}[Temporal motif~\cite{Paranjape-2017-motifs}]\label{def:temporal_motif}
A \emph{$k$-node, $l$-edge} temporal motif $M = (\multigraph, \sigma)$ 
consists of a multigraph $\multigraph = (V, E)$
with $k$ nodes and $l$ edges and an ordering $\sigma$ on the edges of $E$.
\end{definition}
We often find it convenient to represent $(\multigraph, \sigma)$ by an ordered
sequence of edges $(u_1, v_1), (u_2, v_2), \ldots, (u_l, v_l)$.
\Cref{def:temporal_motif} is a template for a temporal graph pattern, and we want
to count how many times the pattern appears in a temporal network.
Furthermore, we are interested in how often the motif occurs within some
time span $\delta$.
A collection of edges in a temporal graph is a
$\delta$-\emph{instance} of a temporal motif $M = (\multigraph, \sigma)$ if it matches the same edge
pattern of the multigraph $\multigraph$, the temporal edges occur in the specified order $\sigma$, and
all of the temporal edges occur within a $\delta$ time window (see \cref{fig:example_C}).
We now formalize this definition.
\begin{definition}[Motif $\delta$-instance~\cite{Paranjape-2017-motifs}]\label{def:instance}
A time-ordered sequence
$S = (w_1, x_1, t_1)$, $\ldots$, $(w_l, x_l, t_l)$ of $l$ unique temporal edges
from a temporal graph $T$
is a $\delta$-\emph{instance} of the temporal motif
$M = (u_1, v_1), \ldots, (u_l, v_l)$ if
\begin{enumerate}
\item There exists a bijection $f$ on the vertices in $M$ such that
$f(w_i) = u_i$ and $f(x_i) = v_i$, $i = 1, \ldots, l$; and
\item The edges all occur within the $\delta$ time span, i.e., $t_l - t_1 \le \delta$.
\end{enumerate}
\end{definition}
\noindent With this definition, motif instances are defined by just the
existence of edges (a general subgraph) and not the non-existence of edges (an
induced subgraph).

We are interested in counting how many motifs appear within a \emph{maximum}
time span of $\delta$ time units.
Our sampling framework will also make use of the actual \emph{duration} of motif instances,
or the difference in the latest and earliest timestamp of a motif instance.
We formalize this notion in the following definition.
\begin{definition}[Motif duration]\label{def:duration}
Let $S = (w_1, x_1, t_1)$, $\ldots$, $(w_l, x_l, t_l)$ be an
instance of a motif $M$ as per \cref{def:instance} with
$t_1 < t_2 < \ldots t_l$.
Then the \emph{duration} of the instance, denoted $\duration{S}$,
is $t_l - t_1$.
\end{definition}


\section{Counting temporal stars is hard}
Star motifs as in \cref{fig:example_B} are one of the fundamental small graph
patterns and are used in, e.g., anomaly detection~\cite{Akoglu-2010-oddball} and
graph summarization~\cite{Koutra-2015-summarizing}. In static graphs, counting
non-induced instances of stars is simple. Given the degree $d_u$ of node $u$,
$u$ is the center of ${d_u \choose k}$ $(k+1)$-node stars. Thus, there is a
simple polynomial-time algorithm for computing the total number of stars.

In contrast, once we introduce temporal information, it turns out that stars
become hard to compute. Specifically, we show in this section that counting
temporal stars is NP-complete, and even determining the existence of a temporal
star motif is NP-complete. This result serves two purposes. First, it highlights
that the computational challenges with temporal graph data are fundamentally
different from those in traditional static graph analysis. Second, the
computational difficulty in such a simple type of temporal motif motivates the
need for scalable approximation algorithms, which we develop in the next
section. We begin with a formal definition of a temporal star motif.
\begin{definition}
  A $k$-\emph{temporal star} is a temporal motif where the multigraph
  is connected and has node set $\{0,1,\ldots,k\}$ with edges
  $(u_i, v_i)$, $i = 1, \ldots, m$, where either $u_i$ or $v_i$ is 0, $i = 1, \ldots, m$.
\end{definition}
The restriction that either $u_i$ or $v_i$ is 0 means that each edge either
originates from node 0 or enters node 0.  The ordering $\sigma$ of the edges in
the multigraph needed by \cref{def:temporal_motif} is arbitrary---we only need
the star structure of the multigraph.
We will show that determining the existence of an instance of a $k$-\emph{temporal star}
in a temporal graph is NP-complete and then generalize our result to an even more restricted class of
star motifs. We begin with the formal decision problem.
\begin{problem}
Given a temporal graph $T$, a $k$-temporal star $S$, and a time span $\delta$,
the \textsc{k-Star-Motif} problem asks if there exists at least one
$\delta$-instance of $S$ in $T$.
\end{problem}

To establish NP-completeness, we reduce \textsc{k-Clique} to
\textsc{k-Star-Motif}. A \textsc{k-Clique} problem instance is formalized as
follows: given an undirected graph $G$ and an integer $k$, the \textsc{k-Clique} problem asks if
there exists at least one clique of size $k$ in $G$.
\begin{theorem}\label{thm:npcomplete}
\textsc{k-Star-Motif} is NP-complete.
\end{theorem}
\begin{proof}
Our input is an instance $(G,k)$ of \textsc{k-Clique} on a vertex set $V$.
Assume that the nodes in $V$ are numbered from $1$ to $n = \lvert V \rvert$ (\cref{fig:construction_A}).
We construct an instance $(T, S, \delta)$ of \textsc{k-Star-Motif}:
\begin{itemize}
\item 
\textbf{Construction of $T$} (\cref{fig:construction_B}).
For each undirected edge $(u,v)$ in $G$, add to $T$ two edges
$(0, u, (v-1)\cdot (n+2) + u + 1)$ and $(0, v, (u-1)\cdot (n+2) + v + 1)$.
For each $u \in V$, we add two backward edges,
$(u, 0, (u-1)\cdot(n+2)+1)$ and $(u, 0, u\cdot (n+2))$.
\item
\textbf{Construction of $S$} (\cref{fig:construction_C}).
For each node $u\in V$, add two backward edges with timestamps $(u-1) \cdot
(n+2) + 1$ and $u \cdot (n+2)$, and $k-1$ forward edges with timestamps
$\left\{(i-1)\cdot (n+2) + 1 + u \,\vert \, i \in [n]\setminus\{u\}\right\}$.
\item Set $\delta = \infty$.
\end{itemize}

\begin{figure}[t]
  \begin{center}
    \phantomsubfigure{fig:construction_A}
    \phantomsubfigure{fig:construction_B}
    \phantomsubfigure{fig:construction_C}    
    \includegraphics[width=\columnwidth]{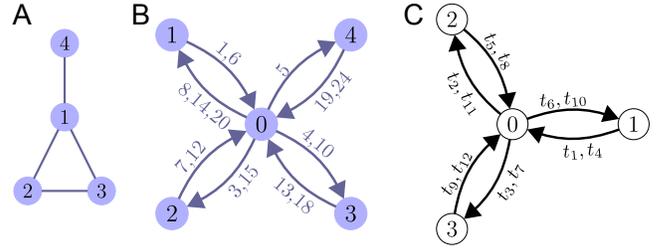}
    \caption{Structures used in proof of \cref{thm:npcomplete}, which says that
    determining the existence of a temporal star is NP-complete.
    (A) A static graph $G$.
    (B) A temporal graph $T$.
    (C) A star motif $S$.
    With the reduction, there is a 3-clique in $G$ if and only if there is a $\delta$-instance of $S$ in $T$
    with $\delta = \infty$.
    }
    \label{fig:construction}
\end{center}
\end{figure}

The timestamps come from the set $\{1, \ldots, n^2 + 2n\}$, and we think of
the timestamps as partitioned into $n$ blocks, with $n+2$ timestamps in each
block. If the timestamp of an edge lies in
$\{(u-1)\cdot (n+2) + 1, \ldots, u\cdot (n+2)\}$,
then we say that the edge belongs to block $u$. Each block then corresponds to
a node in $G$, with the first and last timestamp in each block reserved for the
backward edges we add to $T$. For each node $u$ in the original graph, we add
the two backward edges in block $u$ to node $u$ in $T$, and for each neighbor
$v$ of $u$, we add a forward edge using the timestamp in the $(u+1)$-th position
of block $v$. \Cref{fig:construction} is a schematic of the construction.
Observe that if there is a clique in $G$, then
by construction the star motif $S$ occurs in $T$.

Intuitively, the backward edges added to $T$ and $S$ serve as ``bookends''. If
the two backward edges corresponding to a node $u$ are found to be part of $S$,
then each of the $k-1$ other nodes in the motif has to contribute a forward edge
with timestamps between the two backward edges of $u$. By construction of $S$,
an edge connected to $v$ can only have a timestamp in block $u$ if $v$ is
connected to $u$ in $G$. This implies that $u$ is connected to the $k-1$ other
nodes selected in the motif $S$. Applying this argument to each node $u$ in the
motif $S$, there must be a clique in the original graph $G$.
\end{proof}

The result does not depend on having edges in two directions.  We call a star
motif $S$ \emph{unidirectional} if all of the edges in $S$ either originate from
or enter the center node (node 0, in our notation).
\begin{theorem}
\textsc{k-Star-Motif} is NP-complete even when restricted to unidirectional stars.
\end{theorem}
\begin{proof}
(Sketch.) Instead of using two backward edges for bookkeeping, we can expand the
size of each block to $3n$ and use the first $n$ and last $n$
timestamps within the block as the bookends. Thus, the graph $T$ in the
previous proof is modified by connecting the center node to each node $u$ with $2n$
forward edges using the $2n$ timestamps reserved for bookkeeping
in block $u$. The motif $S$ is modified by requiring the same forward edges as
before, plus an additional $2n$ forward edges, with timestamps in
$3 (u-1)\cdot n \cdot + i$ and $3 (u-1) \cdot n + n + i$ for $i = 1, \ldots, n$.
By using $n$ forward edges for each bookend, we ensure that any occurrence $S$
found in $T$ must include at least one edge from each bookend of the chosen
nodes. This allows us to again argue that $k-1$ forward edges must be between
the bookends of block $u$, implying that there is a $k$-clique.
\end{proof}

These hardness results illustrate the computational difficulties in counting
temporal graph patterns, which motivates scalable approximation algorithms for
counting such patterns. We next present a general sampling framework for
scalable estimation of the number of instances of temporal motifs.


\section{Algorithmic sampling framework}\label{sec:algorithms}

Suppose we are given a motif $M$, a time span $\delta$, and a temporal graph
$T$.
In this section, we develop a sampling framework to estimate the number of
$\delta$-instances of $M$ in $T$, which we denote by $C_M$.
Our sampling framework will employ some algorithm that can compute exactly the
number of $\delta$-instances of $M$ on temporal subgraphs of $T$.
The requirements on the algorithm are that, given a temporal graph $T'$, a motif
$M$, and a time span $\delta$, the algorithm $A$ outputs a sequence of the
count-duration pairs $\{(\text{count}_i, \Delta_i)\}$, where $\text{count}_i$ is
the number of instances of the motif with duration $\Delta_i$.
We denote this output by $A(T', M, \delta)$.
We work from these assumptions in this section, and \cref{sec:experiments}
discusses compatible algorithms.

\xhdr{Intervals and the count vector $Y_s$}
We begin with some definitions and technical lemmas that will later be used to
develop our estimator.
Let $s$ be a random integer uniformly drawn from $\{-c\delta+1,\ldots,0\}$ for some input
integer $c > 0$ that controls the size of the sampling windows.
We call $s$ a \emph{shift}, and we will eventually make use of multiple shifts
within our sampling framework.
We consider the set of intervals of width $c\delta$ with shift $s$:
\begin{equation}
  {\mathcal I}_s = \{[s+(j-1)c\delta, s+j\cdot c\delta-1],\; j=1,2,\ldots\}.
\end{equation}
For an instance $M'$ of the motif $M$ with duration $\duration{M'}$, it is easy
to see that the probability (over a random choice of shift $s$) that $M'$ is
completely contained within an interval in ${\mathcal I}_s$ is
\begin{equation}
p_{M'}=1-\frac{\duration{M'}}{c\delta}.
\end{equation}

Next, for an interval $I \in \mathcal{I}_s$, let $X_{M'}(I)$ be an indicator
random variable which equals 1 if $M'$ is completely contained in $I$ and 0 otherwise.
For each interval $I \in {\mathcal I}_s$, we associate a weighted count $w(I)$
of the number of instances of motif $M$ completely contained in the interval $I$:
\begin{equation}\textstyle
w(I)= \sum_{M'} \frac{1}{p_{M'}} X_{M'}(I).
\end{equation}
Let $Y_s$ be the vector of such counts:
\begin{equation}
Y_{s,j} = w(I_j),\; I_j = [s+(j-1)c\delta, s+j\cdot c\delta-1] \in {\mathcal I}_s
\end{equation}
(here, $Y_{s, j}$ denotes the $j$th coordinate of $Y_s$).
Next, let $X_{M'}$ be an indicator random variable that equals 1 if the motif
instance $M'$ is completely contained in an interval in ${\mathcal I}_s$ and 0
otherwise.
Then $\lVert Y_s \rVert_1 = \sum_{M'} \frac{1}{p_{M'}} X_{M'}$.
The following lemma says that $\lVert Y_s \rVert_1$ is an unbiased estimator for
the motif count $C_M$ for any value of $s$.
\begin{lemma}\label{lem:Ys_unbiased}
$\E[\lVert Y_s \rVert_1] = C_M$.
\end{lemma}
\begin{proof}
%
Since $\E[X_{M'}] = p_{M'}$, 
$\E [\lVert Y_s \rVert_1] = \sum_{M'} \frac{1}{p_{M'}} \E[X_{M'}] = \sum_{M'} 1 = C_M$.
\end{proof}

\noindent The next lemma bounds the variance of $\lVert Y_s \rVert_1$.
\begin{lemma}\label{lem:var_Ys}
$\Var [\lVert Y_s \rVert_1] \leq \frac{1}{c-1} C_M^2$.
\end{lemma}
\begin{proof}
First, we have that
\[ \textstyle
\E [\lVert Y_s \rVert_1^2] = \E \left[\sum_{M'}\frac{1}{p_{M'}} X_{M'} \right] = \sum_{M_1} \sum_{M_2} \frac{1}{p_{M_1}p_{M_2}}\E[X_{M_1} X_{M_2}],
\]
where $M_1$ and $M_2$ range over the instances of the motif $M$.
Using the bounds $\frac{1}{p_{M_2}} \leq \frac{c}{c-1}$ and $\E[X_{M_1} X_{M_2}] \leq \E[X_{M_1}]$,
\[ \textstyle
\E [\lVert Y_s \rVert_1^2] \leq \frac{c}{c-1} \sum_{M_1} \frac{1}{p_{M_1}} \E[X_{M_1}] \sum_{M_2} 1 = \frac{c}{c-1} C_M^2.
\]
Putting everything together,
\begin{align*}
  \Var[\lVert Y_s \rVert_1]
  &= \E [\lVert Y_s \rVert_1^2] - \left(\E [\lVert Y_s \rVert_1]\right)^2 \\
  &\leq \frac{c}{c-1} C_M^2 - C_M^2 = \frac{1}{c-1} C_M^2.
\end{align*}
\end{proof}

Our sampling framework estimates $\lVert Y_s \rVert_1$
in order to estimate the number of motif instances $C_M$.
The basic idea of our approach is to use importance sampling to speed up this
estimation task, by picking a set of intervals in ${\mathcal I}_s$ and computing
their weights.
Here, computing the weight for an interval $I$ uses an \emph{exact} motif count
restricted to the interval $I$.
Equivalently, we (i) sample a subset of coordinates of $Y_s$, (ii)
compute their values exactly, and (iii) combine them to estimate $\lVert Y_s \rVert_1$.
We describe this procedure next.

\xhdr{Importance sampling for an estimator}
Let $Y_{s,j}$ denote the $j$th coordinate of $Y_s$, corresponding to interval
$I_j$.
Our estimator $Z$ is a random variable defined as follows.
First, we sample interval $I_j \in {\mathcal I}_s$ (independently) with some
probability $q_j$.
These $q_j$ values will be based on simple statistics of the intervals; we
will specify choices for $q_j$ later but note for now that they do not
necessarily sum to 1.
Second, let $Q_j$ be an indicator random variable corresponding to interval $I_j$,
where $Q_j$ equals 1 if $j$ is picked and 0 otherwise.
Finally, our estimator is
\begin{align}\label{eq:estimator}
Z \triangleq \sum_j Q_j \frac{Y_{s,j}}{q_j}.
\end{align}

Our first result is that $Z$ is an unbiased estimator for $C_M$,
the number of instances of the motif $M$.
\begin{theorem}
The random variable $Z$ in \cref{eq:estimator} is an unbiased estimator for
the number of motif instances, i.e., $\E[Z] = C_M$.
\end{theorem}
\begin{proof}
First, note that $\E[Q_j] = q_j$.
For any $s$, $\E[Z \given s] = \sum_j Y_{s,j} = \lVert Y_s \rVert_1$.
Hence, $\E[Z] = \E[\lVert Y_s \rVert_1] = C_M$ by \cref{lem:Ys_unbiased}.
\end{proof}

Next, we work to bound the variance of our estimator $Z$.
To this end, it will be useful to define a scaled version $\hat{Y}_s$ of $Y_s$:
\begin{align}
  \hat{Y}_{s,j} \triangleq Y_{s,j} / \sqrt{q_j}.
\end{align}
The following lemma provides a useful equality on the variance of our estimator
in terms of $\hat{Y}_s$ and $Y_s$, conditioned on the shift $s$.
\begin{lemma}\label{lem:var_Zs}
$\Var[Z \given s] = \lVert \hat{Y}_s \rVert_2^2 - \lVert Y_s\rVert_2^2$.	
\end{lemma}
\begin{proof}
By independence of the $Q_j$,
\[ \textstyle
\Var[Z \given s] = \sum_j \Var \left[Q_j \frac{Y_{s,j}}{q_j}\right] = \sum_j \frac{Y_{s,j}^2}{q_j^2} q_j(1-q_j).
\]
Therefore, $\Var[Z \given s] = \sum_j Y_{s,j}^2 / q_j - Y_{s,j}^2 = \lVert \hat{Y}_s\rVert_2^2 - \lVert Y_s\rVert_2^2$.
\end{proof}

\noindent We are now ready to bound the variance of $Z$.
\begin{theorem}
$\Var[Z] \leq \E[\lVert \hat{Y}_s\rVert_2^2] - \E[\lVert Y_s\rVert_2^2] + \frac{1}{c-1} C_M^2$.	
\end{theorem}
\begin{proof}
For this bound, we first condition on $s$ and then take the expectation over random choice of $s$. 
\begin{align*}
  \Var[Z]
  &= \E[(Z-C_M)^2] \\
  &= \E_s\left[\left( (Z - \lVert Y_s \rVert_1) + (\lVert Y_s \rVert_1 - C_M) \right)^2\right] \\
  &= \E_s\left[\Var[Z \given s]+	(\lVert Y_s \rVert_1 - C_M)^2\right]\\
  &= \E[\lVert \hat{Y}_s\rVert_2^2] - \E[\lVert Y_s\rVert_2^2] + \Var[\lVert Y_s \rVert_1] &\text{(by \cref{lem:var_Zs})}\\
  &\leq \E[\lVert \hat{Y}_s\rVert_2^2] - \E[\lVert Y_s\rVert_2^2] + \frac{1}{c-1} C_M^2    &\text{(by \cref{lem:var_Ys})}	
\end{align*}
\end{proof}

Our analysis thus far has been for a single shift $s$.
If we repeat the above computations for $b$ randomly chosen shifts and report
the mean of the estimates, then the variance is reduced by a factor of $b$.
\Cref{alg:sampling} outlines the the overall sampling procedure, assuming that
the sampling probabilities $q_j$ are given along with the exact counting
algorithm $A$. In the algorithm, we use $T_j$ to denote the subgraph restricted
to interval $I_j$ and $A(T_j, M, \delta)$ to denote the output of the exact
counting algorithm on the interval, which is a sequence of the count-duration
pairs $\{(\text{count}_i, \Delta_i)\}$ of motif instances contained in the
interval. The algorithm also explicitly states that the parallelism that can be performed
over the samples.

\begin{algorithm}[tb]
  \DontPrintSemicolon
  \KwIn{Temporal graph $T$, motif $M$, time span $\delta$,
    sampling probabilities $q$, number of shifts $b$,
    window size parameter $c$, exact motif counting algorithm $A$.}
  \KwOut{Estimate of the number of instances of $M$.}
  \caption{Sampling framework for estimating the number of instances of a
    temporal motif in a temporal network. Without loss of generality, the
    timestamps in the temporal network are normalized to start at 0.}
  $Z \leftarrow 0$,\quad $t_{\max} \leftarrow \max \{ t \given (u, v, t) \in T\}$\;
  \For{$k = 1, \ldots, b$}{
    $s \leftarrow$ random integer from $\{-c\delta+1, \ldots, 0\}$\;
    \For{$j = 1, \ldots, 1+\lceil \frac{t_{\max}}{c\delta} \rceil$ (in parallel)}{
      \If{$\text{Uniform}(0,1) \leq q_j$}{
        $T_j \leftarrow \{(u, v, t) \in T \given t \in [s+(j-1)c\delta, s+j\cdot c\delta-1]\}$\;
        \For{$(\textnormal{count}_i, \Delta_i) \in A(T_j, M, \delta)$}{
          $Z_k \leftarrow Z_k + \text{count}_i / ((1 -\Delta_i / (c\delta)) \cdot q_j)$\;
        }
      }
    }
  }
  \KwRet{$\frac{1}{b}\sum_{k=1}^{b}Z_k$}
  \label{alg:sampling}
\end{algorithm}

\xhdr{Choosing the sampling probabilities}
In order to get average squared error $(\epsilon C_M)^2$, we need to set the
parameters as follows:
\begin{align}
\E[\lVert \hat{Y}_s\rVert_2^2] - \E[\lVert Y_s\rVert_2^2] + \frac{1}{c-1} C_M^2 &\leq b(\epsilon C_M)^2\\
\iff \frac{\E[\lVert \hat{Y}_s\rVert_2^2] - \E[\lVert Y_s\rVert_2^2]}{C_M^2} + \frac{1}{c-1} &\leq b\epsilon^2. \label{eq:tradeoffs}
\end{align}
The first term in the left-hand side of \cref{eq:tradeoffs}
combines (i) a natural measure of sparsity of the distribution of motifs with
(ii) the extent of correlation between the sampling probabilities $q_j$ and the
(weighted) motif counts for intervals $Y_{s,j}$.
In order to understand this, let $\ell$ denote the dimension of $Y_s$ and
consider the simple uniform setting of $q_j = 1/\ell$ (so one interval is
sampled in expectation). In this case, the term becomes
\begin{equation}\label{eq:sparsity_measure}
\frac{(\ell-1)\E[\lVert Y_s\rVert_2^2]}{\E[\lVert Y_s\rVert_1^2]}.
\end{equation}

\Cref{eq:sparsity_measure} is a natural measure of sparsity of the vector $Y_s$.
In the extreme case where $Y_s$ is a vector with only one non-zero coordinate,
the value is $\ell-1$, and in the other extreme where $Y_s$ is a uniform vector,
the value is bounded above by 1.
In the sparse case, we need to increase the sampling probabilities---thus
sampling more intervals---to compensate for the large variance.
In the worst case, this would require looking at all the intervals, i.e., we get
no running time savings from sampling (however, we will see in our experiments
that the data is far from the worst case in practice).
Nonetheless, the ability of the algorithm to pick sampling probabilities $q_j$ gives
flexibility to mitigate the dependence on the sparsity of $Y_s$.
To illustrate this point, in the extremely favorable case when $q_j$ is
proportional to $Y_{s,j}$, i.e., $q_j = Y_{s,j} / \sum_j Y_{s,j}$ (so one
interval is sampled in expectation), the first term on the left-hand side of \cref{eq:tradeoffs} is
less than 1.
This analysis suggests that a good choice of sampling probabilities roughly
balances the two terms:
\[
\frac{\E[\lVert \hat{Y}_s\rVert_2^2] - \E[\lVert Y_s\rVert_2^2]}{C_M^2} \approx \frac{1}{c-1}.
\]

A priori, we do not know $Y_s$ or $\hat{Y}_s$. What we can do is choose the
sampling probabilities by some easily measured statistic that we think is
correlated with $Y_s$, such as the number of temporal edges or number of static
edges. For this paper, we simply choose $q_j$ to be proportional to the number
of temporal edges in the interval, i.e.,
\begin{equation}\label{eq:heuristic_q}
q_j = r \cdot \frac{\lvert \{ (u, v, t) \in T \given t \in I_j \} \rvert}{\lvert T \rvert},
\end{equation}
where $r$ is a small constant (in practice on the order of 10--100).
This leads to substantial speedups, as we will see in the next section.
There are certainly more sophisticated approaches one could take to choose the
$q_j$, and we leave this as an avenue for future research.

\xhdr{Streaming from sampling}
When memory is at a premium, the sampling framework above can be made memory
efficient. By considering the windows ${\mathcal I}_s$ in chronological order,
the edges of past windows do not need to be stored. By running several
estimators in parallel, we can achieve any accuracy we want while only needing
to store edges in an interval of at most $c\delta$ at a time. As we will see in
our experiments, the memory savings allows us to processes larger temporal
graphs than we could with an exact algorithm.


\section{Computational experiments}\label{sec:experiments}
In this section, we use our sampling framework from \cref{sec:algorithms} and
various exact temporal motif counting algorithms to count temporal motifs on
real-world datasets. By exploiting sampling and the ability sample in parallel,
we obtain substantial speedups with modest computational resources and only a
small error in the estimation.

\xhdr{Data}
We gathered 10 datasets for our experiments. Paranjape et al.\ analyzed
seven of them~\cite{Paranjape-2017-motifs}, and we collected three larger datasets
to better analyze the performance of our methodology. \Cref{tab:datasets}
lists summary statistics of the datasets, and we briefly describe them
below. Each dataset is a collection of timestamped directed edges. The time
resolution of each dataset is 1 second, except for the EquinixChicago dataset,
where the time resolution is 1 microsecond.

\begin{table}[tb]
  \setlength{\tabcolsep}{3pt}
  \caption{Summary statistics of temporal networks.}
  \label{tab:datasets}
  \begin{tabular}{r @{\quad} c c c c}
    \toprule
    dataset         & \# nodes & \# static & \# temporal & time  \\
                    &          & edges           & edges   & span \\\midrule
    CollegeMsg      & 1.9K  & 20.3K & 59.8K & 194 days   \\
    email-Eu-core   & 986   & 24.9K & 332K  & 2.20 years \\
    MathOverflow    & 24.8K & 228K  & 390K  & 6.44 years \\
    AskUbuntu       & 157K  & 545K  & 727K  & 7.16 years \\
    SuperUser       & 192K  & 854K  & 1.11M & 7.60 years \\
    WikiTalk        & 1.09M & 3.13M & 6.10M & 6.24 years \\
    StackOverflow   & 2.58M & 34.9M & 47.9M & 7.60 years \\
    Bitcoin         & 48.1M & 86.8M & 113M  & 7.08 years \\
    EquinixChicago & 12.9M & 17.0M & 345M  & 62.0 mins  \\
    RedditComments & 8.40M & 517M  & 636M  & 10.1 years \\
    \bottomrule
  \end{tabular}
\end{table}

\noindent \emph{CollegeMsg}~\cite{Panzarasa-2009-patterns}.
A network of private messages sent on an online social network at the University
of California, Irvine.
\newline
\noindent \emph{email-Eu-core}~\cite{Paranjape-2017-motifs}.
A collection of internal email records from a European research institution.
\newline
\noindent \emph{MathOverflow, AskUbuntu, SuperUser, and StackOverflow}~\cite{Paranjape-2017-motifs}.
These datasets are derived from user interactions on Stack Exchange question and
answer forums. A temporal edge represents a user replying to a question,
replying to a comment, or commenting on a question.
\newline
\noindent \emph{WikiTalk}~\cite{Leskovec-2010-Governance,Paranjape-2017-motifs}.
A network of Wikipedia users making edits on each others' ``talk pages.''
\newline
\noindent \emph{Bitcoin}~\cite{Kondor-2014-inferring}.
A network representing timestamped transactions on Bitcoin. The addresses were
partially aggregated by a de-identification heuristic~\cite{Reid-2012-analysis}
implemented by Kondor et al., using all transactions up to February 9,
2016~\cite{Kondor-2014-inferring}. Timestamps are the creation time of the block
on the blockchain containing the transaction.  We will release this dataset with
the paper.
\newline
\noindent \emph{EquinixChicago}~\cite{CAIDA-data}.
This dataset was constructed from passive internet traffic traces from CAIDA's
monitor in Chicago on February 17, 2011. Each edge represents a packet sent from
one anonymized IP address to another. Data was collected from the ``A
direction'' of the monitor.
\newline
\noindent \emph{RedditComments}~\cite{Hessel-2016-science}.
This dataset was constructed from a large collection of comments made by users
on \url{https://www.reddit.com}, a popular social media platform. A comment from
user $u$ to user $v$ at time $t$ induces a temporal edge in our dataset.


\subsection{Exact counting algorithms}\label{sec:exact_counting}
Our sampling framework is flexible since it can use any algorithm that exactly
counts temporal motifs as a subroutine, provided that this algorithm can be
transformed to output the count-duration pairs $\{(\text{count}_i, \Delta_i)\}$,
where $\text{count}_i$ is the number of instances of the motif with duration
$\Delta_i$. A recently proposed ``backtracking'' algorithm satisfies this
constraint~\cite{Mackey-2018-chronological}. The fast algorithms for 2-node,
3-edge star motifs introduced by Paranjape et al.\ do not satisfy these
requirements, since the algorithm uses an inclusion-exclusion rule that cannot
output the durations. However, we still use this algorithm as a baseline in our
experiments. We also create a new exact counting algorithm that is compatible
with our sampling framework, which we describe below.

\xhdr{Backtracking algorithm (BT,~\cite{Mackey-2018-chronological})}
The backtracking algorithm examines the edges of the input graph in chronlogical
order and matches one edge of the motif at a time. The software was not released
publicly, so we have re-implemented it with some optimizations. The algorithm
is compatible with our sampling framework. The algorithm is inherently
sequential, so our parallel sampling framework is especially useful with this
method.

\xhdr{Fast 2-node, 3-edge algorithm (F23,~\cite{Paranjape-2017-motifs})}
Paranjape et al.\ introduced a collection of algorithms for counting motifs with
at most 3 edges. Here we use their specialized algorithm for motifs with two
nodes and three edges, which produces exact counts in time linear in the number
of edges. This algorithm is incompatible with our sampling framework, but we
still use it for comparative purposes.

\begin{algorithm}[tb]
  \DontPrintSemicolon
  \KwIn{Two nodes $u$ and $v$, and a sequence of temporal
    edges $(e_1, t_1), \ldots,$ $(e_L, t_L)$ with $t_1 < \ldots < t_L$, time
    span $\delta$, and $e_i = (u,v)$ or $(v, u)$.}
  \KwOut{List $\{(\text{count}_i, \Delta_i)\}$ of counts of instances
    of the motif in \cref{fig:experimental_motifs_A} between nodes $u$ and $v$ with durations
    $\Delta_i$.}
  \caption{EX23: A simple exact algorithm to count the two-node motif in
    \cref{fig:experimental_motifs_A}. This algorithm can
    easily be modified to count any 2-node, 3-edge temporal motif.}
  $C \leftarrow$ empty counter dictionary with default value 0\;
  \For{$i = 1 \ldots L$}{
    \lIf{$e_i \neq (u, v)$} {continue}
    $N_b \leftarrow 0$\;
    \For{$j = i+2 \ldots L$} {
      $\Delta \leftarrow t_j - t_i$\;
      \lIf{$t_j - t_i > \delta$} {break}
      \lIf{$e_j = (v, u)$} {$N_b \leftarrow N_b + 1$}\lElse
      {
        $C[\Delta] \leftarrow C[\Delta] + N_b$\;
      }
    }
  }
  \KwRet{\textnormal{[$(C[\Delta], \Delta)$ for key $\Delta$ in $C$]}}
\label{alg:EX23}
\end{algorithm}

\begin{table*}[t]
\caption{%
Running time in seconds of algorithms with and without our sampling framework and with and
without parallelism.
The modifiers ``+S'' and ``+PS'' stands for sampling and parallelized sampling
respectively.
We compare the backtracking algorithm (BT,~\cite{Mackey-2018-chronological}),
our specialized algorithm for counting 2-node, 3-edge motifs (EX23,
\cref{alg:EX23}), and the specialized fast algorithm from Paranjape et al.\ for
counting 2-node, 3-edge motifs (F23,~\cite{Paranjape-2017-motifs}).
\emph{In all datasets, our EX23 algorithm within our parallel sampling framework has the
  fastest running time.}
}
\label{tab:perf23}
\begin{tabular}{r @{\qquad} c c c @{\qquad} c c c @{\qquad} c c @{\qquad} c}
\toprule
 dataset        & BT    & BT+S  & BT+PS & EX23  & EX23+S & EX23+PS & F23   & F23+P & error (\%) \\ \midrule
 CollegeMsg     & 0.076 & 0.072 & 0.038 & 0.017 & 0.016  & 0.009   & 0.056 & 0.054 & 0.33       \\
 Email-Eu       & 0.339 & 0.305 & 0.191 & 0.073 & 0.078  & 0.027   & 0.217 & 0.165 & 1.44       \\
 MathOverflow   & 0.545 & 0.361 & 0.143 & 0.233 & 0.148  & 0.097   & 0.998 & 0.878 & 1.74       \\ 
 AskUbuntu      & 1.414 & 1.305 & 0.500 & 0.592 & 0.311  & 0.176   & 2.534 & 2.371 & 1.84       \\ 
 SuperUser      & 2.590 & 1.446 & 0.483 & 1.097 & 0.194  & 0.104   & 4.595 & 4.129 & 1.69       \\ 
 WikiTalk       & 15.92 & 14.88 & 5.463 & 4.737 & 3.645  & 0.876   & 20.46 & 18.23 & 0.89       \\ 
 StackOverflow  & 198.9 & 160.8 & 79.58 & 108.1 & 69.50  & 17.81   & 299.2 & 230.1 & 1.95       \\ 
 Bitcoin        & 514.0 & 520.7 & 102.3 & 494.4 & 233.5  & 88.66   & 10348 & 10135 & 3.59       \\ 
 EquinixChicago & 480.4 & 180.3 & 37.64 & 382.7 & 56.33  & 24.64   & 477.3 & 383.8 & 0.00       \\ 
 RedditComments & 7301  & 7433  & 2910  & 1563  & 3154   & 367.4   & 6602  & 5036  & 4.83       \\ \bottomrule
\end{tabular}
\end{table*}

\xhdr{A new exact algorithm for 2-node 3-edge motifs (EX23)}
We devised a new algorithm for 2-node, 3-edge motifs that is compatible with our
sampling framework. In this case, each pair of nodes in the input graph forms
an independent counting problem. For each pair of nodes in the input that are
neighbours in the static graph, we gather all temporal edges between the two
nodes. Then we fix the first and last edge of the 3-edge motif by iterating over
all pairs of gathered temporal edges. By maintaining an additional counter of
the number of edges between the two fixed edges, we can count the number of
temporal motifs that begins on the first fixed edge and ends on the second fixed
edge (the procedure is outlined in \cref{alg:EX23}). Overall, this procedure
takes $O(\sum_{u,v}k_{u,v}^2)$ time, where the sum iterates over all pairs of
nodes in the graph, and $k_{u,v}$ is the number of temporal edges between nodes
$u$ and $v$. With additional code complexity from special tree structures, the
running time can be improved to $O(\sum_{u,v}k_{u,v}\log k_{u,v})$ and still be
compatible with our sampling framework. However, this optimization is not
crucial for the main focus of our paper, which is the acceleration of counting
algorithms with sampling. Thus, we use the simpler un-optimized algorithm, which
we will see actually out-performs the other exact counting algorithms.

\subsection{Performance results}
We now evaluate the performance of several algorithms:
(i) the three baseline exact counting algorithms described in the previous section (BT, F23, EX23);
(ii) the F23 baseline with parallelism enabled (F23+P);
(iii) our sampling framework on top of backtracking and our new exact counting algorithm (BT+S, EX23+S); and
(iv) our parallelized sampling framework on top of backtracking and our new exact counting algorithm (BT+PS, EX23+PS).
As explained above, the F23 algorithm is incompatible with our sampling
framework; we include the algorithm and its parallelized version as baselines
for fast exact counting.

All algorithms were implemented in C++, and all experiments were executed on a
16-core 2.20 GHz Intel Xeon CPU with 128 GB of RAM. The algorithms ran on a
single thread unless explicitly stated to be parallel. The parallel algorithms
used 16 threads. In the case of the sampling algorithm, parameters are set so
that the approximations are within 5\% relative error of the true value.

\begin{figure}[tb]
  \centering
  \phantomsubfigure{fig:experimental_motifs_A}
  \phantomsubfigure{fig:experimental_motifs_B}
  \includegraphics[width=0.8\columnwidth]{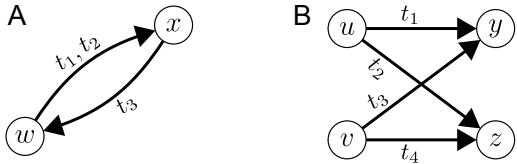}
  \caption{Motifs used in counting experiments. (A) The 2-node 3-edge motif for which
    results are reported in \cref{tab:perf23}. (B) The bi-fan motif for which results
    are reported in \cref{tab:bifan}.}
  \label{fig:experimental_motifs}
\end{figure}


\xhdr{Experiments on a 2-node, 3-edge motif}
\Cref{tab:perf23} reports the performance of all algorithms on the 2-node,
3-edge temporal motif in \cref{fig:experimental_motifs_A} (we chose this motif
to allow us to compare against one of the fast algorithms of Paranjape et
al.~\cite{Paranjape-2017-motifs}). The time span $\delta$ was set to 86400
seconds = 1 day in all datasets except EquinixChicago, where $\delta$ was 86400
microseconds (these are the same parameters used in exploratory data analysis in
prior work~\cite{Mackey-2018-chronological}).

We highlight three important findings.
First, our new EX23 algorithm with parallel sampling is the fastest algorithm on
every dataset. Comparing our algorithm against the previous state of the art, we
see speedups up to 120 times faster than the slowest exact algorithm (see the
results for Bitcoin). This is in part due to the fact that our EX23 algorithm
is actually faster the than the backtracking algorithm (BT) and
the fast algorithm of Paranjape et al.\ (F23). In other words, our proposed
exact algorithm already out-performs the current state of the art.

Second, in all cases, parallel sampling provides a substantial speedup over
the exact baseline algorithm. Speedups are typically on the order of 2--6x
improvements in running time. We used 16 threads but did not optimize our
parallel algorithms; there is ample room to improve these results with
additional software effort.

Third, the running time of the backtracking algorithm with sampling (BT+S) is often
comparable to simple backtracking (BT). In these cases, we hypothesize that the
backtracking algorithm has enough overhead and is pruning enough edges to make
simple sampling not worthwhile under our parameter settings. However, parallel
sampling with the backtracking algorithm (BT+PS) can yield substantial speedups
(see, e.g., Bitcoin, EquinixChicago, and RedditComments). This
illuminates an important feature of our sampling framework, namely, we
get parallelism for free. The backtracking algorithm is inherently sequential,
but parallel sampling can achieve substantial speedups. Thus, future research in
the design of fast exact counting algorithms can largely leave parallelism to be
handled by our sampling framework. Finally, although not reported, the sampling
framework requires a smaller amount of memory than the exact algorithms; thus,
if no parallelism is available, we can at least gain in terms of memory, if not
in speed.

\begin{figure}[t]
\includegraphics[width=\columnwidth]{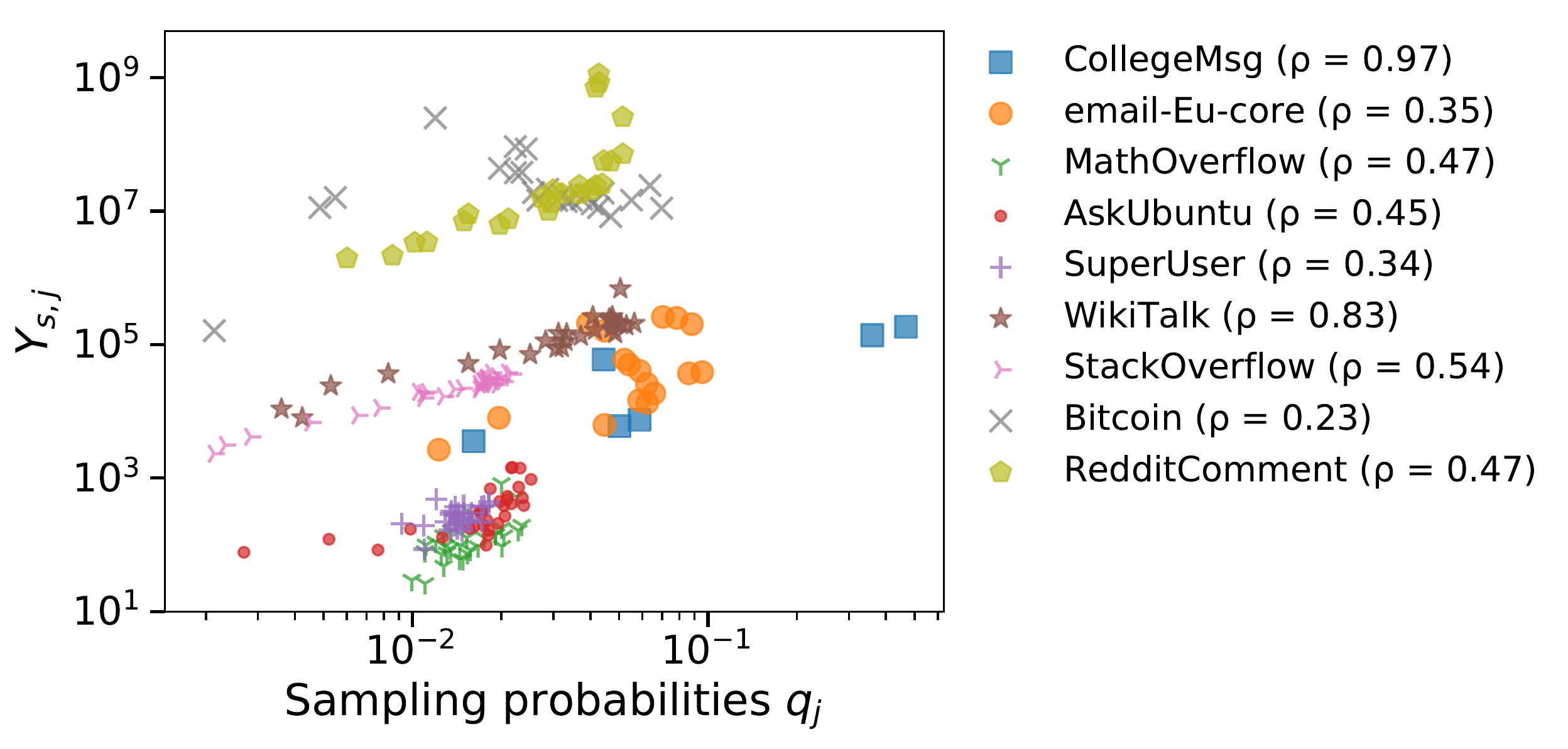}
\caption{%
Sampling probabilities, which are based on the number of edges
(\cref{eq:heuristic_q}) and the values of $Y_s$ for one sample $s$ used in
estimating the counts of the 2-node, 3-edge temporal motif in
\cref{fig:experimental_motifs_A}. For our sampling procedure to be effective,
these values should be positively correlated (see the discussion in
\cref{sec:algorithms}). Positive correlation leads to less variance in the
estimation. The correlation $\rho$ is listed in the legend. The CollegeMsg and
WikiTalk datasets both have a strong positive correlation and a small relative
error in the motif count estimate (see \cref{tab:perf23}).
}
\label{fig:correlations}
\end{figure}

We used the heuristic in \cref{eq:heuristic_q} to determine the sampling
probabilities $q$ in these experiments. To understand why this heuristic worked
for these datasets (i.e., the relative errors are small), we measured the
correlation of $q$ and the coordinates of the vector $Y_s$ used in the sampling
framework (\cref{fig:correlations}). In datasets such as CollegeMsg and
WikiTalk, the correlation between is large, and consequently, the
relative error in the estimates is small (\cref{tab:perf23}).

\xhdr{Experiments on a 4-node, 4-edge bi-fan motif}
Next, we show the results of the backtracking algorithm on a so-called
``bi-fan'' motif (\cref{fig:experimental_motifs_B}). This motif has four nodes
and four temporal edges. The static version of the motif appears is important in networks
from a variety of domains, including sociology~\cite{Zhang-2013-potential},
neuroscience~\cite{Benson-2016-hoo}, gene
regulation~\cite{Dobrin-2004-aggregation}, and circuit
design~\cite{Milo-2002-motifs}. Since this motif has four nodes and four edges,
our new fast algorithm and the fast algorithm of Paranjape et al.\ cannot be
used. Thus, we focus on accelerating the backtracking algorithm with
(parallelized) sampling. For these experiments, we set $\delta = 3600$, as
running the algorithm with $\delta = 86400$ exceeds the alloted memory on the
Bitcoin and RedditComments dataset. \Cref{tab:bifan} shows the results.

\newcommand{\nofinish}{\xmark}
\begin{table}[t]
\caption{%
Running times in seconds of the backtracking algorithm
(BT,~\cite{Mackey-2018-chronological}) with our sampling framework (``+S''
denotes serial sampling and ``+PS'' denotes parallel sampling).
The symbol ``\xmark'' indicates that the algorithm failed due to
the machine running out of memory.
Our sampling framework uses less memory than the exact counting algorithm, so it
is always able to estimate the motif count on these datasets. 
}
\label{tab:bifan}
\begin{tabular}{r @{\qquad} c c c c}
\toprule
 dataset        & BT        & BT+S  & BT+PS & error (\%) \\ \midrule
 CollegeMsg     & 0.081     & 0.076 & 0.069 & 1.02       \\
 Email-Eu       & 0.353     & 0.307 & 0.120 & 0.85       \\
 MathOverflow   & 0.528     & 0.362 & 0.041 & 3.60       \\ 
 AskUbuntu      & 1.408     & 0.909 & 0.078 & 4.52       \\ 
 SuperUser      & 2.486     & 1.269 & 0.164 & 2.33       \\ 
 WikiTalk       & 51.85     & 35.35 & 11.99 & 2.01       \\ 
 StackOverflow  & 221.7     & 93.10 & 5.208 & 4.88       \\ 
 Bitcoin        & 1175      & 985.9 & 269.3 & 3.09       \\ 
 EquinixChicago & 481.2     & 45.50 & 5.666 & 1.33       \\ 
 RedditComments & \nofinish & 6739  & 2262  & --         \\ \bottomrule
\end{tabular}
\end{table}

Again, the parallel sampling procedure provides a substantial speed-up over the
baseline algorithm. We emphasize that our simple parallelization technique is a
property of the sampling procedure and not a property of the exact algorithm. In
fact, the exact algorithm is inherently sequential, and the sampling framework
enables parallelism in a trivial way with minimal loss in accuracy. Moreover,
since the sampling algorithm only examines a portion of the graph at a time, it
uses much less memory than the exact counting algorithm. For example, with the
RedditComments dataset, the exact algorithm ran out of memory, while the
sampling algorithms completed successfully (thus no relative error is
reported). This feature is useful in streaming applications, where memory is
limited.


\section{Additional related work}\label{sec:related}

Our sampling framework relies on importance sampling, which is used for finding
motifs in gene sequence
analysis~\cite{siddharthan2008phylogibbs,chan2010importance,gupta2007variable,liang2008sequential};
here ``motif'' refers to short string patterns in DNA. (The term ``motif'' in
the context of network analysis is borrowed from this
domain~\cite{shen2002network}.)  We have already covered much of the related
work in sampling algorithms for pattern counting in static graphs, so we
summarize additional research related to various definitions of temporal motifs here.
Some of these are for sequences of static snapshot
graphs~\cite{Zhang-2014-dynamic,Lahiri-2007-structure,Jin-2007-trend}, which is
a different data model than the one in this paper, where edges have timestamps
in a continuum.  For the data in our paper, there are motifs based on ``adjacent
events'' that require each new edge in a sequence to be within a certain
timespan of each
other~\cite{Zhao-2010-communication-motifs,Gurukar-2015-commit,Hulovatyy-2015-exploring}.
These definitions are slightly more restrictive than the one by Paranjape et
al.\ analyzed here; however, the principles of our techniques could also be
adapted to these cases as well.  Kovanen et al.\ use the same notion of event
adjacency but also restrict motif instances to cases where the events are
consecutive for all nodes involved (i.e., within the span of the motif instance,
there can be no other temporal edge adjacent to one of the
nodes)~\cite{Kovanen-2011-motifs}.  This definition is even more restrictive in
the events that it captures but it does allows for much faster exact counting
algorithms, e.g., triangle motifs can be counted in linear time in the size of
the data. Thus, speeding up computation with sampling is less appealing for this
definition.  Finally, there is also a line of research in finding dense
subgraphs in datasets similar to our model, which is a specific type of
motif~\cite{Gaumont-2016-dense,Viard-2018-enumerating,Viard-2016-computing}.


\section{Discussion}
We have developed a sampling framework for estimating the number of instances of temporal motifs
in temporal graphs.
Overall, our sampling framework is flexible in several ways. First, the
framework is built on top of exact counting algorithms. Improvements in these
algorithms can be used directly within our framework, provided it meets the
conditions needed by our framework. In fact, we demonstrated this to be the case with
the specialized algorithm we developed for 2-node, 3-edge motifs in
\cref{sec:experiments}, which was faster than existing exact
counting methods for that particular motif. Second, our sampling framework
provides natural parallelism, which allowed us to achieve
orders of magnitude speedups on algorithms that do not have obvious parallel
implementations. Finally, the sampling is inherently less memory intensive,
which allowed us to estimate motif counts on datasets on which exact algorithms
cannot even run (\cref{tab:bifan}); thus, our framework makes knowledge
discovery feasible in new cases.

An extraordinary amount of research has gone into scalable
estimation algorithms for counting patterns in static graphs.
Our paper takes this line of research in a new direction by considering richer
patterns that arise when temporal information is incorporated into the graph.
We anticipate that our work will open new challenges for algorithm designers
while simultaneously providing a solution for domain scientists
working with large-scale temporal networks.

\section*{Acknowledgements}
This research was supported in part by NSF Award DMS-1830274. We thank Eugene Y.Q. Shen for donating the computing resources that made some of the larger datasets possible.

\bibliography{refs}
\bibliographystyle{abbrv}

\end{document}